\documentclass[11pt,reqno]{amsart}
\usepackage{fullpage}
\usepackage{amsthm,amsmath,amsfonts,amssymb}
\usepackage{hyperref}
\usepackage{enumerate}
\usepackage{bbm}
\usepackage{xcolor}

\usepackage{tikz}
\usetikzlibrary{decorations.pathmorphing,patterns,decorations.markings,matrix,quantikz2,arrows,shadows}
\usepackage{hyperref}
\usepackage[capitalize]{cleveref}
\usepackage{bbold,enumitem,mathtools}
\usepackage{physics}
\usepackage{booktabs}
\usepackage{adjustbox}

\newcounter{nodenumber}

\newcounter{arraycard}
\newcommand{\drawchainwithfields}[3][0.4\textwidth]{
\def\firstlist{#2}
\def\secondlist{#3}
\begin{center}
\begin{adjustbox}{width=#1}
\begin{tikzpicture}
\setcounter{arraycard}{0}
  \foreach \x [count=\c] in \firstlist {%
    \stepcounter{arraycard}
    \draw [thick] (2*\c,0) -- node [midway,above]{\x} (2*\c+2,0) ;
  }
  \foreach \y [count=\c] in \secondlist
  {%
    \node[circle,fill=black,minimum size=0.8cm,label=below:\y] at (2*\c,0) {};
  }

\end{tikzpicture}
\end{adjustbox}
\end{center}
}

\newcommand{\identity}{\mathbb{1}}

\newcommand{\ignore}[1]{}
\newcommand{\fixme}[1]{{\color{black} #1}}
\newcommand{\etal}{{\em et al.}~}
\newcommand{\iverson}[1]{\lbrack\!\lbrack #1 \rbrack\!\rbrack}

\newtheoremstyle{prenum}
  {5pt}
  {5pt}
  {\slshape}
  {0pt}
  {\bfseries}
  {.}
  { }
  {\thmnumber{#2}\thmname{ #1}\thmnote{ (#3)}}

\newtheorem{definition}{Definition}
\newtheorem{theorem}{Theorem}[section]

\newtheorem{lemma}[theorem]{Lemma}

\newtheorem{claim}{Claim}[section]

\newif\ifnotesw\noteswtrue
   {\ifnotesw\marginpar[\hfill\(\top\)]{\(\top\)}\fi}%
   {\ifnotesw\marginpar[\hfill\(\bot\)]{\(\bot\)}\fi}

\newcommand{\mnote}[1]%
    {\ifnotesw\marginpar%
        [{\scriptsize\begin{minipage}[t]{\marginparwidth}
        \raggedleft#1%
                        \end{minipage}}]%
        {\scriptsize\begin{minipage}[t]{\marginparwidth}
        \raggedright#1%
                        \end{minipage}}%
    \fi}

\newcommand{\ZZ}{\mathbb{Z}}
\newcommand{\RR}{\mathbb{R}}
\newcommand{\CC}{\mathbb{C}}

\newcommand{\OO}{\mathcal{O}}

\DeclareMathOperator{\Sp}{Spec}

\DeclareMathOperator{\nul}{Null}
\DeclareMathOperator{\ran}{Ran}

\DeclareMathOperator{\diag}{diag}
\DeclareMathOperator{\Mat}{Mat}

\DeclareMathOperator{\myspan}{span}
\DeclareMathOperator{\Adj}{Adj}

\newcommand{\Exp}{\mathbb{E}}

\newcommand{\lip}[2]{\langle #1, #2 \rangle}

\newcommand{\trex}{{T}.{\em rex}~}

\newcommand{\heff}{\mathbb{h}}

\title{The strength of weak coupling}
\author{Alastair Kay}
\address{Department of Mathematics, Royal Holloway University of London, Egham, Surrey, TW20 0EX, UK.}
\email{alastair.kay@rhul.ac.uk}
\author{Christino Tamon}
\address{Department of Computer Science, Clarkson University, Potsdam, New York 13699, USA.}
\email{ctamon@clarkson.edu}

\date{\today}
\begin{document}
\maketitle

\begin{abstract}
A paradoxical idea in quantum transport is that attaching weakly-coupled edges to a large base graph
creates high-fidelity quantum state transfer.
We provide a mathematical treatment that rigorously prove this folklore idea.
Our proofs are elementary and build upon the Feshbach-Schur method from perturbation theory.
We also show the idea is effective in circumventing Anderson localization in spin chains
and finding speedups in hitting times useful for quantum search.
\end{abstract}

\section{Introduction}

A continuous-time quantum walk on a graph $G$ is a time-varying unitary matrix $U(t)=\exp(-itA)$ 
where $A$ is the adjacency matrix of $G$.
This framework is useful for developing quantum algorithms (see \cite{fg98,fg98b}), 
which had culminated in a quantum walk algorithm with provable exponential speedup 
\cite{ccdfgs03} (a guarantee missing in both Shor factoring and Grover search).

For two vertices $a,b$ of $G$ with characteristic vectors $e_a,e_b$, let their fidelity function be $f_{a,b}(t)=|\lip{e_b}{U(t)e_a}|$.
By Born's rule, the probability of measuring $b$ at time $t$ given that 
the quantum walk starts at $a$ is $f_{a,b}(t)^2$.
We say $G$ has pretty good state transfer (PGST) from $a$ to $b$ if for any $\epsilon > 0$
there exists a time $t_0$ so that $f_{a,b}(t_0) \ge 1-\epsilon$ (see \cite{gkss12,vz12}); 
it is called perfect state transfer (PST) if $\epsilon=0$ is realizable.
Quantum state transfer is fundamental to quantum communication as it can be used as a building block 
for protocols such as entanglement generation and state preparation (see \cite{cdel04,christandl2005,acde04,ks05,k10}).

Although PST/PGST are useful abstractions, there are several important obstacles in developing
state transfer in practice. Arguably, PST is impossible to realize in real-world applications.
Recently, due to experimental challenges in realistic scenarios, 
a stronger emphasis has been placed on {\em high-fidelity} state transfer (HFST) (see \cite{cmf16,ls24,ls25,vz24,cgs}).
Here, for a small but fixed $\epsilon_0 > 0$, we ask for an efficient time $t_0$ for which $f_{a,b}(t_0) \ge 1-\epsilon_0$.
This differs markedly from PGST as $\epsilon_0$ is not required to be arbitrarily close to $0$
and time $t_0$ must be a polynomial function in the relevant parameters (such as graph size or diameter or transfer distance).
Here, the emphasis is on transfer time which is neither nondeterministic nor astronomical.

We explore a paradoxical idea (called \trex) which attaches weakly-coupled pendant edges (or tiny arms) to a base graph (large body).
This idea was studied by W\'{o}jcik \etal (see \cite{wlkggb05,wlkggb07}).
Here, we describe a mathematical treatment of this idea using the Feshbach-Schur method in perturbation theory (see \cite{k80,gs20}).
Within this framework, we rigorously prove that high-fidelity quantum state transfer occurs.
The advantage of using the Feshbach-Schur method is that it is elementary 
and it insulates asymptotic arguments (the bane of perturbation theory).
We also prove that the transfer time of the \trex scheme does not depend on the diameter.
In contrast, methods based on strong potentials exhibit exponential scaling with the diameter (see \cite{ls24,ls25}).

Another formidable obstacle for state transfer is Anderson localization (see \cite{a58,klmw07,redsm16}).
By definition, localization is the absence of transport across a spin chain due to random diagonal perturbations. 
We show how the \trex idea overcomes this obstacle by constructing a robust state transfer protocol which has high fidelity.
Our protocol exploits nonsingularities of the underlying spin chains in a direct way along with the use of a few clean qubits.

We also apply the \trex idea to hitting times on graphs with consequences to quantum search.
For the latter, we reimagine the vertex-located oracle (see \cite{fg98}) as a pendant edge oracle.
The search algorithm simply ``pitches'' a random pendant edge onto the graph and apply the \trex walk.
This reduces quantum search, viewed as reverse mixing, to a quantum hitting process (which is a standard
technique in Markov chain theory \cite{lpw}).

As concluding remarks, the \trex idea had already found two unexpected applications elsewhere:
a simple matrix inversion quantum algorithm \cite{kt} 
and an optimal {\em perfect} state transfer protocol robust against timing sensitivities \cite{kkt}. 
We anticipate further uses of the \trex idea which we leave to future treks.

\section{Preliminaries}

We review notation and terminology used in this work.
Let $\Mat_n(\RR)$ be the set of all $n \times n$ real matrices.
We use $\identity_n$ and $J_n$ to denote the identity and all-one matrices of order $n$, respectively;
we omit the subscript, whenever the dimensions are clear from context.
The spectrum of $A \in \Mat_n(\RR)$ is denoted $\sigma(A)$. In most cases, we will restrict ourselves to symmetric matrices, whose spectrum exists and whose eigenvalues are real.

For an invertible matrix $A \in \Mat_n(\RR)$ 
with $\lambda_{\max}(A) = \max\{|\lambda| : \lambda \in \sigma(A)\}$
and $\lambda_{\min}(A) = \min\{|\lambda|: \lambda \in \sigma(A)\}$, 
the {\em condition number} of $A$ is defined as
\[
	\kappa(A) = \norm{A} \norm{A^{-1}} = \frac{\lambda_{\max}(A)}{\lambda_{\min}(A)}
\]
where $\norm{A}$ is the spectral norm of $A$ (see \cite{hj2}, page 385).
We shall assume $\norm{A}$ is bounded.

For a matrix $A \in \Mat_n(\RR)$, the {\em resolvent} of $A$ is given by 
\begin{equation} \label{eqn:resolvent}
R(\zeta,A) = (\zeta \identity - A)^{-1}.
\end{equation}
Its resolvent set $\rho(A)$ is the set of all values $\zeta \in \CC$ for which $\zeta\identity - A$ is invertible (and bounded);
note that $\sigma(A) = \CC \setminus \rho(A)$.
We state two useful identities involving the resolvent.

\begin{lemma} (Teschl \cite{t14})
For matrices $A,B \in \Mat_n(\RR)$ and for $\zeta,\mu \in \CC$, we have
\begin{align} 
\label{eqn:1st-resolvent}
R(\zeta,A) - R(\mu,A) &= (\mu - \zeta)R(\zeta,A)R(\mu,A),
	\ \ \mbox{ $\zeta,\mu \in \rho(A)$} \\
\label{eqn:2nd-resolvent}
R(\zeta,A) - R(\zeta,B) &= R(\zeta,A)(A - B)R(\zeta,B),
	\ \ \mbox{ $\zeta \in \rho(A) \cap \rho(B)$.}
\end{align}
\end{lemma}

For $z \in \CC$ and a subset $S \subset \CC$, the distance of $z$ to $S$ is defined as
$d(z,S) = \inf\{|z - a| : a \in S\}$.

\begin{lemma} \label{lemma:resolvent-bound}
For a matrix $A \in \Mat_n(\RR)$ and a scalar $\zeta \not\in \sigma(A)$, we have
$\norm{R(\zeta,A)} = 1/d(\zeta,\sigma(A))$.

\begin{proof}
Note $\norm{(\zeta\identity - A)^{-1}} = \sup\{|\zeta - \lambda|^{-1} : \lambda \in \sigma(A)\}$.
The latter equals to $1/\inf\{|\zeta - \lambda| : \lambda \in \sigma(A)\}$, proving the claim.
\end{proof}
\end{lemma}

For simplicity, we sometimes identify a projection operator $P$ with its range $\ran(P)$.

\section{The Feshbach-Schur method}

Let $H \in \Mat_n(\RR)$ be a symmetric matrix and let $P_0,P_1$ be a pair of orthogonal projections 
satisfying $P_0 + P_1 = I$. 
Clearly, $H$ admits a block decomposition given by
\[
	H = (P_0 + P_1)H(P_0 + P_1) = P_0 H P_0 + P_0 H P_1 + P_1 H P_0 + P_1 H P_1.
\]

\par\noindent{\bf Notation}: Let $H_{j,k} = P_j H P_k$, for $j,k=0,1$, denote the projected operators (acting on the full space). On the other hand, we let $H_{P_j}$ denote 
$P_j H P_j$ when restricted to $\ran(P_j)$.

To motivate the ensuing discussion, suppose we wish to solve the eigenvalue-eigenvector equation given by
\begin{equation} \label{eqn:eigen-problem}
	\begin{pmatrix} 
	A_{0} & A_{1} \\
	A_{2} & A_{3}
	\end{pmatrix}
	\begin{pmatrix}
	\psi_0 \\
	\psi_1
	\end{pmatrix}
	=
	\lambda
	\begin{pmatrix}
	\psi_0 \\
	\psi_1
	\end{pmatrix}.
\end{equation}
This implies 
\begin{align*}
A_{0}\psi_0 + A_{1}\psi_1 &= \lambda\psi_0 \\
A_{2}\psi_0 + A_{3}\psi_1 &= \lambda\psi_1.
\end{align*}
The second equation yields
\begin{align} \label{eqn:psi1}
	\psi_1 &= (\lambda\identity - A_{3})^{-1}A_{2}\psi_0
\end{align}
which we can substitute into the first to obtain
\[
	\left( A_{0} + A_{1}(\lambda\identity - A_{3})^{-1}A_{2} \right) \psi_0 = \lambda\psi_0.
\]
Notice this is related to Schur complementation.

Based on these ideas, if we let
\begin{equation} \label{eqn:feshbach}
	F(\lambda) = [H + H P_1(\lambda\identity - H_{P_1})^{-1} P_1 H]_{P_0},
\end{equation}
we may reduce the eigenvalue-eigenvector problem 
$H\psi = \lambda\psi$ to a fixed-point equation
\begin{equation} \label{eqn:fixed-point-problem}
	F(\lambda)\psi = \lambda\psi
\end{equation}
provided $\lambda\identity - H_{P_1}$ is invertible in $\ran(P_1)$.
Solving for $\lambda$ and $\psi$ in \cref{eqn:fixed-point-problem}, 
the resulting eigenvector of $H$ is 
\[
	\Psi = Q(\lambda)\psi
\]
where
\[
	Q(\lambda) = P_0 + P_1(\lambda \identity - H_{P_1})^{-1} P_1 H P_0.
\]
From the notation, it will be clear from context whether the operator acts
on the full or projected spaces.

This reduction of the eigenvalue problem of $H$ in \cref{eqn:eigen-problem}
to a nonlinear problem involving $F(\lambda)$ in \cref{eqn:fixed-point-problem} is known as the {\em Feshbach-Schur} method 
(see \cite{dss21,gs20,sz07}).

\subsection{The \trex Cometh} \label{subsec:trex}
We focus on {\em linear} perturbation theory (see \cite{k80}, page 72) where
\[
	H=H_0+\delta W,
	\ \hspace{.3in} \
	\mbox{ where $\delta\kappa \ll 1$.}
\]
Here $\kappa$ will be the condition number of the relevant base graph (see below).
Let $\lambda_0$ be an eigenvalue of $H_0$ of multiplicity $d$, for $d \ge 2$.
Let $P_0$ be the projection onto $N_0 = \nul(\lambda_0 \identity - H_0)$ 
and $P_1 = \identity - P_0$ be the projection onto the orthogonal complement $N_0^\perp$.

The map $F(\lambda)$ of $H$ relative to $P_0$ is defined as
\begin{equation} \label{eqn:fs-def}
	F(\lambda) := \lambda_0\identity + \delta W_{P_0} + \delta^2 W_{01}(\lambda \identity - H_{P_1})^{-1} W_{10},
\end{equation}
under the assumption that $\lambda\identity - H_{P_1}$ is invertible in $\ran(P_1)$.
    
\begin{theorem} \label{thm:fs}
(Gustafson-Sigal \cite{gs20} and Dusson \etal \cite{dss21})\\ 
The map $F(\lambda)$ satisfies
\begin{equation} \label{eqn:fs-prop}
	H\Psi = \lambda\Psi
	\ \mbox{ if and only if } \
	F(\lambda)\psi = \lambda\psi,
	\ \hspace{0.25in} \ \mbox{ for $\psi \in \ran(P_0)$}
\end{equation}
where vectors $\psi$ and $\Psi$ are related through
\begin{equation} \label{eqn:fs-evecs}
\Psi = Q(\lambda)\psi 
	\ \ \mbox{ and } \ \
\psi = P_0\Psi
\end{equation}
with
\begin{equation} \label{eqn:fs-lift}
	Q(\lambda) = P_0 + \delta P_1(\lambda \identity - H_{P_1})^{-1}W_{10}.
\end{equation}
Thus, $\Psi = \psi + \psi^\perp$, for orthogonal $\psi$ and $\psi^\perp$,
where 
\begin{equation} \label{eqn:state-perp}
\psi^\perp = \delta P_1(\lambda \identity - H_{P_1})^{-1} P_1 W\psi
\end{equation}
\end{theorem}

Suppose $\zeta_k \in \RR$ and $\psi_k \in \CC^n$ are fixed-point solutions of \cref{eqn:fs-prop}: 
\[
	F(\zeta_k)\psi_k = \zeta_k\psi_k,
	\ \hspace{.2in} \
	\mbox{ for $k=1,\ldots,d$.}
\]

\par\noindent{\bf Assumption A}: The scalars $\zeta_k$ are {\em distinct}.
In the parlance of perturbation theory, the degeneracy of $\lambda_0$ is {\em fully resolved}.

By \cref{thm:fs}, each $\zeta_k$ is an eigenvalue of $H$ with eigenvector
\[
	\Psi_k = \psi_k + \psi_k^\perp,
	\ \ \
	\mbox{ where $\psi_k^\perp = \delta P_1 R(\zeta_k, (H_{11})_{P_1}) P_1 W \psi_k$.}
\]

\begin{lemma} \label{lemma:eigval-estimate} (Gustafson and Sigal \cite{gs20}, page 138) \\
Let $\lambda$ be a fixed-point solution of $F(\lambda)$. Then,
\[
	\lambda = \lambda_0 + \delta\mu(\delta) + \OO(\delta^3),
\]
where $\mu(\delta) \in \sigma(M)$ with
\[
M = W_{P_0} + \delta [W_{01} R(\lambda_0, (H_{11})_{P_1})W_{10}]_{P_0}.
\]
\end{lemma}

\cref{lemma:eigval-estimate} can be applied to a more general setting than what is considered here 
(see \cite{gs20} for further details).

\smallskip
\par\noindent{\bf Assumption B}: 
We make the following additional assumptions:
\begin{enumerate}[label = (\roman*)]
\item {\em $W$ acts trivially in the $P_0$ and $P_1$ subspaces}.\\
	That is, $W_{00} = W_{11} = 0$. In effect, $W$ connects the subspaces $P_0$ and $P_1$. 

\item {\em $\lambda_0 = 0$}.\\
	If not, we may always take the diagonal shift $H_0 - \lambda_0 I$.

\item {\em $(H_0)_{P_1}$ is (sufficiently) nonsingular or has good condition number}.\\
	This implies that for a (small enough) $\kappa$, we have
	\begin{equation}
	d(0,\sigma((H_0)_{P_1})) \ge 1/\kappa.
	\end{equation}
	But, in some cases, we can remove this assumption (see \cref{sec:resonant}).
\end{enumerate}
Using the third assumption, by \cref{lemma:resolvent-bound}, 
we have 
\[
	\norm{R(0, (H_0)_{P_1})} \le \kappa.
\]
From \cref{lemma:eigval-estimate}, we have $\mu(\delta) \le \delta\kappa$ which implies
\begin{equation} \label{eqn:eigval-bound}
	|\lambda| = \OO(\delta^2\kappa).
\end{equation}
Furthermore, 
\begin{align} 
d(\lambda, \sigma((H_0)_{P_1})) 
	&\ge d(0, \sigma((H_0)_{P_1})) - |\lambda| \\
	&\ge 1/\kappa - \OO(\delta^2\kappa).
\end{align}
As $\delta\kappa \ll 1$, we see that
\begin{equation} \label{eqn:aux-bound}
d(\lambda, \sigma((H_0)_{P_1})) = \Omega(1/\kappa).
\end{equation}

\subsection{Effective Hamiltonian}

In what follows, let $\varepsilon := \delta\kappa$.
To track our quantum walk within the subspace $P_0$, we use the effective Hamiltonian (defined below)
which we will derive from the (full) perturbed Hamiltonian $H$.

Assume that the spectral decomposition of $H$ is given by
\[
	H = \sum_{k=1}^{d} \zeta_k \Psi_k^{\dagger}\Psi_k + \fixme{\sum_\theta \theta Q_\theta}
\]
where $\zeta_k$ are the $d$ distinct eigenvalues (split off from the eigenvalue $\lambda_0$ by perturbation)
and $\theta$ represent the other eigenvalues (with eigenprojectors $Q_\theta$).
Furthermore, let
\[
	\Psi_k = \psi_k + \psi_k^\perp
	\ \hspace{.5in} \
	(k=1,\ldots,d) 
\]
where $\psi_k \in P_0$ are the fixed-point solutions of the Feshbach map (see \cref{thm:fs}).

\begin{claim} \label{claim:relative}
For $k=1,\ldots,d$, $\norm{\psi_k^\perp}/\norm{\psi_k} \le \OO(\varepsilon)$,
whence $\norm{\psi_k}^2 \ge 1-\OO(\varepsilon^2)$ and $\norm{\psi_k^\perp}^2 \le \OO(\varepsilon^2)$.

\begin{proof}
First, we show $\norm{\psi_k^\perp}/\norm{\psi_k} \le \OO(\varepsilon)$; that is, the length of
$\psi_k^\perp$ is negligible compared to the length of $\psi_k$. To see this, note
\begin{align*}
\lVert{\psi_k^\perp}\rVert 
	&\le \delta \norm{P_1 R(\zeta_k,H_{11}) P_1 W} \norm{\psi_k} \\
	&\le \delta \norm{R(\zeta_k,P_1 H_0 P_1)} \norm{\psi_k},
		\ \ \ \mbox{ as $\norm{P_1},\norm{P_1 W} \le 1$ } \\
	&\le \delta \frac{1}{d(\zeta_k, \sigma(P_1 H_0 P_1))} \norm{\psi_k},
		\ \ \ \mbox{ by \cref{lemma:resolvent-bound}} \\
	&\le \OO(\varepsilon) \norm{\psi_k},
		\ \ \ \mbox{ using \cref{eqn:aux-bound}.}
\end{align*}
Now, we have
$1 = \norm{\Psi_k}^2 = \norm{\psi_k}^2 + \lVert\psi_k^\perp\rVert^2 \le (1 + \OO(\varepsilon^2))\norm{\psi_k}^2$.
\end{proof}
\end{claim}

Let $\tilde{\psi}_k = \psi_k/\norm{\psi_k}$ be the state $\psi_k$ normalized. 
We shall show $\{\tilde{\psi}_k : k=1\ldots,d\}$ is a nearly-orthonormal basis for $P_0$.

\begin{claim} \label{claim:orthogonal}
If $k \neq \ell$, $\lip{\tilde{\psi}_k}{\tilde{\psi}_\ell} \le \OO(\varepsilon^2)$.

\begin{proof}
We have
$\lip{\Psi_k}{\psi_k^\perp} = \norm{\psi_k^\perp}^2 \le \OO(\varepsilon^2)$ by \cref{claim:relative}.
Thus,
\[
	\lip{\psi_k}{\psi_\ell}
	= \lip{\Psi_k - \psi_k^\perp}{\Psi_\ell - \psi_\ell^\perp} 
	= -\lip{\Psi_k}{\psi_\ell^\perp}-\lip{\psi_k^\perp}{\Psi_\ell} + \lip{\psi_k^\perp}{\psi_\ell^\perp} 
	\le \OO(\varepsilon^2)
\]
because
$|\lip{\psi_k^\perp}{\psi_\ell^\perp}| \le \norm{\psi_k^\perp}\norm{\psi_\ell^\perp} \le \OO(\varepsilon^2)$,
by \cref{claim:relative}.
Now, observe $\norm{\psi_k}^2 \ge 1-\OO(\varepsilon^2)$, again by \cref{claim:relative}. 
\end{proof}
\end{claim}

Let $V$ be a matrix whose $k$th column is the vector $\tilde{\psi}_k$, $k=1,\ldots,d$.
Consider the Gram matrix $M = V^\dagger V$ whose entries are $M_{jk} = \lip{\tilde{\psi}_j}{\tilde{\psi}_k}$.

\begin{claim} \label{claim:span}
If $d\varepsilon^2 = o(1)$, $\myspan\{\tilde{\psi}_1,\ldots,\tilde{\psi}_d\} = P_0$.

\begin{proof}
By \cref{claim:orthogonal}, $M_{k,k}=1$ and
$|M_{j,k}| \le \OO(\varepsilon^2)$. If $(d-1)\varepsilon^2 = o(1)$, $M$ is strictly diagonally dominant and
hence invertible (see \cite{hj2}, Theorem 6.1.10). Thus, $M$ has rank $n$, and so does $V$. 
As each $\tilde{\psi}_k$ lies in $P_0$ (which has dimension $d$), we are done.
\end{proof}
\end{claim}

\begin{claim} \label{claim:negligible-overlap}
For any normalized $\psi \in P_0$,
$\norm{\sum_\theta Q_\theta\psi}^2 \le \OO(d\varepsilon^2)$ for each $\theta$.

\begin{proof}
By \cref{claim:span}, $\psi = \sum_k c_k\tilde{\psi}_k$, for some $c_k$, or
$\psi = Vc$, where $c$ is the vector whose $k$th entry is $c_k$, $k=1,\ldots,d$.
Note $c \neq 0$ as $\psi \in P_0$.
As $\psi$ is normalized,
\[
	1 = \norm{\psi}_2^2 = c^\dagger V^\dagger Vc = c^\dagger Mc.
\]
But, $M = I + A$, where $|A_{jk}| = \OO(\varepsilon^2)$ for $j \neq k$, and $A_{jj}=0$, by \cref{claim:orthogonal}.
Therefore,
\[
	1 = c^\dagger Mc 
	= \norm{c}_2^2 + c^\dagger Ac
	= \norm{c}_2^2 \left(1 + \frac{c^\dagger Ac}{\norm{c}_2^2}\right).
\]
By Rayleigh's principle, the ratio $c^\dagger Ac/\norm{c}_2^2$ is bounded in absolute value by $\norm{A}_2$,
which is $\OO(d\varepsilon^2)$.
Hence, $\norm{c}_2^2 = 1 - \OO(d\varepsilon^2)$.
By Cauchy-Schwarz, $\sum_k |c_k| \le \sqrt{d}\norm{c}_2 \le \sqrt{d}$.

Since $\Psi_k = \psi_k + \psi_k^\perp$ and $\sum_\theta Q_\theta\Psi_k = 0$, we have
$\sum_\theta Q_\theta\psi_k = -\sum_\theta Q_\theta\psi_k^\perp$.
But, $\norm{\sum_\theta Q_\theta\psi_k^\perp} \le \norm{\psi_k^\perp} \le \OO(\varepsilon)$ by \cref{claim:relative},
as $\sum_\theta Q_\theta$ is a projection operator.
This shows $\norm{\sum_\theta Q_\theta\psi_k} \le \OO(\varepsilon)$.
Therefore,
\[
	\norm{\sum_\theta Q_\theta\psi} \le \sum_k |c_k|\norm{\sum_\theta Q_\theta\tilde{\psi}_k} \le \sqrt{d}\OO(\varepsilon).
\]
This shows $\norm{\sum_\theta Q_\theta\psi}^2 \le \OO(d\varepsilon^2)$.
\end{proof}
\end{claim}

\begin{claim} \label{claim:tail-bound}
For any normalized states $\psi_a,\psi_b \in P_0$, we have
$|\lip{\psi_b}{\sum_\theta e^{i\theta}Q_\theta\psi_a}| \le \OO(\sqrt{d}\varepsilon)$.

\begin{proof}
Let $\tilde{Q} = \sum_\theta e^{i\theta}Q_\theta$ and note $\tilde{Q}^\dagger\tilde{Q} = \sum_\theta Q_\theta$.
We have
\[
	|\lip{\psi_b}{\tilde{Q}\psi_a}|^2 
	\le \norm{\psi_b}^2 \norm{\tilde{Q}\psi_a}^2 \\
	\le \lip{\psi_a}{\tilde{Q}^\dagger\tilde{Q}\psi_a} \\
	\le \lip{\psi_a}{\sum_\theta Q_\theta\psi_a}.
\]
As $\sum_\theta Q_\theta$ is a projector, the last expression equals $\norm{\sum_\theta Q_\theta\psi_a}^2$
which is at most $\OO(d\varepsilon^2)$ by \cref{claim:negligible-overlap}.
\end{proof}
\end{claim}

\par\noindent
To obtain an orthonormal basis, we apply Gram-Schmidt to the vectors $\tilde{\psi}_1,\ldots,\tilde{\psi}_d$.
Suppose $\rho_1,\ldots,\rho_d$ are the resulting vectors (which span $P_0$).
Note that 
\[
	\rho_k = \tilde{\psi}_k + \OO_d(\varepsilon^2).
\]
The notation $\OO_d(\varepsilon^2)$ hides constant factors that depend on $d$;
if $d = \OO(1)$, this is simply $\OO(\varepsilon^2)$.

The {\em effective Hamiltonian} $\heff$ (relative to $P_0$) is defined as
\begin{equation} \label{eqn:reduced-ham}
\heff := \sum_{k=1}^{d} \zeta_k \rho_k^\dagger \rho_k.
\end{equation}
For two normalized states $\phi_a,\phi_b \in P_0$, we have
\begin{align} 
\lip{\phi_b}{e^{-itH}\phi_a} 
	&= \sum_{k=1}^{d} e^{-it\zeta_k} \lip{\phi_b}{\Psi_k} \lip{\Psi_k}{\phi_a} 
		+ \lip{\phi_b}{\sum_\theta e^{-it\theta} Q_\theta\phi_a} \\
	&= \sum_{k=1}^{d} e^{-it\zeta_k} \lip{\phi_b}{\Psi_k} \lip{\Psi_k}{\phi_a} + \OO_d(\varepsilon),
		\ \ \mbox{ by \cref{claim:tail-bound}}.
\end{align}
Since $\phi_a \in P_0$ and $\Psi_k = \psi_k + \psi_k^\perp$, we get
\[
	\lip{\Psi_k}{\phi_a} 
	= \lip{\psi_k}{\phi_a}
	= \norm{\psi_k} \lip{\tilde{\psi}_k}{\phi_a}
\]
and similarly for $\phi_b$.
Therefore,
\begin{align}
\lip{\phi_b}{e^{-itH}\phi_a} 
	&= \sum_{k=1}^{d} e^{-it\zeta_k} \norm{\psi_k}^2 
		\lip{\phi_b}{\tilde{\psi}_k} \lip{\tilde{\psi}_k}{\phi_a} + \OO_d(\varepsilon) \\
	&= \sum_{k=1}^{d} e^{-it\zeta_k} \norm{\psi_k}^2 
		\lip{\phi_b}{\rho_k} \lip{\rho_k}{\phi_a} + \OO_d(\varepsilon) \\
\label{eqn:qwalk-on-h}
	&= (1-\OO(\varepsilon^2)) \lip{\phi_b}{e^{-it\heff}\phi_a} + \OO_d(\varepsilon)
\end{align}
as $\norm{\psi_k}^2 = 1-\Omega(\varepsilon^2)$ by \cref{claim:relative}.

So, if there is state transfer from $\phi_a$ to $\phi_b$ via $\heff$, then
there is state transfer between the two states via $H$, or
\begin{equation} \label{eqn:good-fidelity}
	|\lip{\phi_b}{e^{-itH}\phi_a}| \ge 1-\OO_d(\varepsilon).
\end{equation}
The fidelity is arbitrarily close to $1$ because $\varepsilon = \delta\kappa = o(1)$.

This shows state transfer on $H$ can be reduced to a (smaller) state transfer on $\heff$.
Unfortunately, the location of each perturbed eigenvalue $\zeta_k$ (relative to $\lambda_0$) 
is only known {\em asymptotically}. Perfect (and pretty good) state transfer requires exact conditions 
on the eigenvalues and eigenprojectors (see \cite{cg}).
We will show how to handle these problems for two-fold degeneracy.

\section{Constructing High-fidelity State Transfer}

The following notion of high-fidelity state transfer is relevant to our results.

\begin{definition} (High-fidelity state transfer)\\
Let $G$ be a graph with adjacency matrix $A(G)$. Let $a,b$ be two vertices of $G$.
We say $G$ has {\em high-fidelity state transfer} from $a$ to $b$ if for any $\epsilon \in (0,1)$
there is an effectively computable time $t=t(\epsilon)$ so that
\begin{equation}
|\bra{b}e^{-iA(G)t}\ket{a}| \ge 1-\epsilon.
\end{equation}
\end{definition}

In most cases, $t(\epsilon)$ is a polynomial function in
$1/\epsilon$ which may also depend on some properties of the underlying graph $G$ (through $A(G)$, say).
We will adopt the convenient Dirac notation which is standard in quantum information.
Thus, $\ket{\alpha}$ denotes a normalized column vector, 
while the row vector $\bra{\alpha} = \ket{\alpha}^\dagger$ denote its conjugate transpose. 
The dimensions will be clear from context.
For $\delta > 0$ and $x \in \RR$, let $B_\delta(x)$ denote the interval $\{y : |y-x| < \delta\}$.

\begin{definition} \label{def:cospectral} ($\gamma$-cospectrality)\\
Let $G$ be a graph with a nonsingular adjacency matrix $A$.
Two normalized vectors $\alpha,\beta \in \CC^{V(G)}$ are {\em $\gamma$-cospectral} in $G$ if there exists $\eta > 0$, with $B_\eta(0) \subset \rho(G)$, so that
\[
	|\bra{\alpha}R(\lambda,A)\ket{\alpha} - \bra{\beta}R(\lambda,A)\ket{\beta}| 
		= \gamma_0 |\bra{\alpha}R(\lambda,A)\ket{\beta}|,
        \ \ \ \mbox{for all $\lambda \in B_\eta(0)$}
\]
for some $\gamma_0 \le \gamma$.
\end{definition}

For a graph $G$, we use $R(\lambda,G)$ in place of $R(\lambda,A)$.
Recall that we say two vertices $a,b$ in $G$ are {\em cospectral} if their walk counts match, that is,
$\bra{a}A^k\ket{a} = \bra{b}A^k\ket{b}$ for all integers $k \ge 1$. It is known that cospectrality is 
a necessary condition for perfect (and pretty good) state transfer between $a$ and $b$ (see \cite{cg}).
Our notion of $\gamma$-cospectrality is weaker than cospectrality,
as $\bra{\alpha}A^k\ket{\alpha} = \bra{\beta}A^k\ket{\beta}$, for $k \ge 0$,
implies $\bra{\alpha}R(\lambda,G)\ket{\alpha} = \bra{\beta}R(\lambda,G)\ket{\beta}$, 
for $\lambda \in \rho(G)$.

We state our main result in this section.

\begin{theorem} \label{thm:high-fidelity}
Let $G_0$ be a connected graph with adjacency matrix $A$, with $\norm{A}=1$, and condition number $\kappa < \infty$.
Let $G = G_0 \cup \overline{K}_2$ be a disjoint union with adjacency matrix $H_0$.
Suppose $\ket{\alpha},\ket{\beta}$ are $(2\gamma)$-cospectral vectors in $G_0$ 
and $\delta \in (0,1)$ satisfy 
\begin{equation} \label{eqn:weak-coupling}
\delta\kappa \ll 1
\end{equation}
and
\begin{equation} \label{eqn:strong-inverse}
|\bra{\alpha}A^{-1}\ket{\beta}| \gg \delta^2\kappa^3.
\end{equation}
If $a,b$ are the vertices of $\overline{K}_2$, 
let $W = \delta(\ketbra{a}{\alpha}+\ketbra{\alpha}{a}+\ketbra{b}{\beta}+\ketbra{\beta}{b})$.
Then, 
\begin{align}
	|\bra{b}e^{-i\tau (H_0 + \delta W)}\ket{a}| &= \frac{1-o(1)}{\sqrt{1+\gamma^2}} - o(1)
\end{align}
where
\begin{equation} \label{eqn:time}
\tau = \OO\left(\frac{1}{\delta^2|\bra{\alpha}A^{-1}\ket{\beta}|}\right).
\end{equation}
\end{theorem}

\begin{proof}
The adjacency matrix of the disjoint union $G=G_0\cup \overline{K}_2$ is $H_0 = O_2 \oplus A$.
The Hamiltonian $H = H_0 + \delta W$ encodes a weak coupling 
between the isolated vertices $a,b$ and the base graph $G_0$.

By Equation \eqref{eqn:weak-coupling}, $d(0,\sigma(G_0)) \ge 1/\kappa$.
Using \cref{lemma:resolvent-bound} and \cref{eqn:aux-bound}, 
for all $\lambda \not\in \sigma(G_0)$ we obtain
\begin{equation} \label{eqn:resolvent-bound}
\norm{R(\lambda,G_0)} \le \kappa.
\end{equation}
The two-dimensional $0$-eigenspace of $H_0$ is spanned by the characteristic vectors $\ket{a}$ and $\ket{b}$.
Let $P_0 = \ketbra{a}{a}+\ketbra{b}{b}$ be the projection onto $\nul(H_0)$ and $P_1 = I - P_0$.

The effective Hamiltonian $\heff = F(\lambda)$ (or the Feshbach-Schur map) is given by
\begin{align*}
F(\lambda) &= \delta^2 W_{01} (\lambda I - A)^{-1} W_{10} \\
	&= \delta^2 (\ketbra{a}{\alpha}+\ketbra{b}{\beta}) (\lambda I - A)^{-1} (\ketbra{\alpha}{a}+\ketbra{\beta}{b}) \\
	&= \delta^2 
		\begin{pmatrix}
		\bra{\alpha}R(\lambda, G_0)\ket{\alpha} & \bra{\alpha}R(\lambda, G_0)\ket{\beta} \\
		\bra{\beta}R(\lambda, G_0)\ket{\alpha} & \bra{\beta}R(\lambda, G_0)\ket{\beta} 
		\end{pmatrix},
		\ \ \ \mbox{ expressed in the $\{\ket{a},\ket{b}\}$ basis.}
\end{align*}
The diagonal entries of $F(\lambda)$ are equal to within $2\epsilon$ due to $(2\epsilon)$-cospectrality 
while the off-diagonal entries are equal since $A(G_0)$ is symmetric.
Furthermore, by a diagonal scaling (which only introduces an irrelevant global phase factor) we may assume that
\begin{equation} \label{eqn:trex-cometh}
F(\lambda) 
	= \delta^2\omega(\lambda) \begin{pmatrix} \gamma & 1 \\ 1 & -\gamma \end{pmatrix},
	\ \hspace{0.2in} \
	\mbox{ where $\omega(\lambda) := |\bra{\alpha}R(\lambda,G_0)\ket{\beta}|$.}
\end{equation}
Up to the $\delta^2\omega(\lambda)$ factor, $F(\lambda) = \gamma Z + X$, where $X,Z$ are the Pauli matrices. 
Therefore,
\[
	e^{-it(\gamma Z + X)} 
	= 
	\cos(t\sqrt{1+\gamma^2})\identity - \frac{i}{\sqrt{1+\gamma^2}} \sin(t\sqrt{1+\gamma^2}) (\gamma Z + X).
\]
Thus, if we take
\begin{equation} \label{eqn:sin-fidelity}
\tau = \frac{\pi}{2} \frac{1}{\delta^2 \omega(\lambda)},
\end{equation}
we get
\begin{equation} 
	\bra{b}e^{-i\tau F(\lambda)}\ket{a}
	=
	-\frac{i}{\sqrt{1+\gamma^2}} \sin(\tau\sqrt{1+\gamma^2}).
\end{equation}
As the exact location of $\lambda$ is unknown, we use $\lambda_0=0$ to approximate it.
Notice that
\begin{align*}
	|\omega(\lambda) - \omega(0)|
		&= |\bra{\beta}(R(\lambda,G_0) - R(0,G_0))\ket{\alpha}| \\
		&\le \norm{R(\lambda,G_0) - R(0,G_0)}, \ \ \ \mbox{ property of the norm} \\
		&= |\lambda|\norm{R(\lambda,G_0)}\norm{R(0,G_0)}, \ \ \ \mbox{ by the first resolvent formula} \\
		&= \OO(|\lambda|\kappa^2), \ \ \ \mbox{ as each resolvent is at most $\kappa$ by \cref{eqn:resolvent-bound}.}
\end{align*}
By an application of \cref{lemma:eigval-estimate}, we have $|\lambda| \le \delta^2\kappa$
which further implies 
\begin{equation} \label{eqn:resolvent-gap}
	|\omega(\lambda) - \omega(0)| \le \OO(\delta^2\kappa^3).
\end{equation}
Thus, in \cref{eqn:sin-fidelity}, we take instead
\begin{equation} \label{eqn:pst-time}
	\tau_0 := \frac{\pi}{2} \frac{1}{\delta^2\omega(0)}.
\end{equation}
Notice that 
\[	
	\tau_0 = \tau(1 + o(1)) 
\]
which follows from $\omega(\lambda)/\omega(0) = 1+o(1)$ because
\[
	\left|\frac{\omega(\lambda)}{\omega(0)} - 1\right| 
	= \frac{|\omega(\lambda) - \omega(0)|}{|\omega(0)|} 
	= o(1),
\]
by \cref{eqn:resolvent-gap} and \cref{eqn:strong-inverse}.
Therefore, we have
\begin{equation} \label{eqn:fidelity}
|\bra{b}e^{-i\tau\heff}\ket{a}| 
	= \frac{1}{\sqrt{1+\gamma^2}} \left|\sin\left(\frac{\pi}{2} (1 + o(1))\right)\right| 
	= \frac{1}{\sqrt{1+\gamma^2}} - o(1). 
\end{equation}
Combining this with \cref{eqn:qwalk-on-h}, we get
\[
	|\bra{b}e^{-i\tau H}\ket{a}| \ge \frac{1-o(1)}{\sqrt{1+\gamma^2}} - o(1)
\]
which completes the proof since $\delta\kappa \ll 1$.
\end{proof}

The explicit appearance of $\bra{\alpha}A^{-1}\ket{\beta}$ in \cref{thm:high-fidelity} was 
the main inspiration behind the simple quantum algorithm for matrix inversion described in \cite{kt}.

We show an illustration of the \trex effects for state transfer on the paths $P_n$ in \cref{fig:p57}.
For further examples and analyses, see \cref{sec:hit}.

\begin{figure}[t]
\includegraphics[scale=0.5]{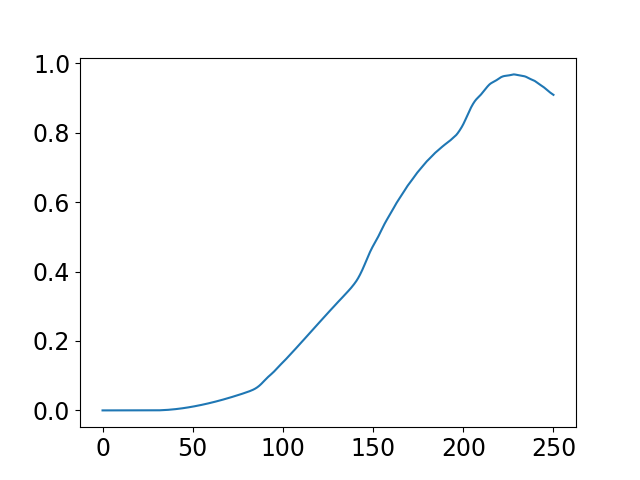}
\includegraphics[scale=0.5]{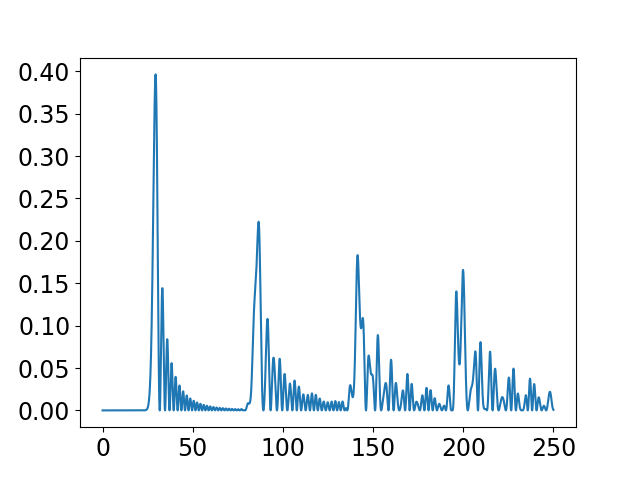}
\caption{The power of weak coupling (\trex effects) in quantum state transfer on $P_{n}$. 
Here, $n=55$ as an example and we plot time $t$ (on X axis) versus antipodal fidelity (on Y axis).
{\em Left}: \trex arms with strength $\delta = 0.05$.
{\em Right}: uniform couplings (fidelity did not reach $0.40$; there is no PGST \cite{gkss12}).}
\label{fig:p57}
\end{figure}

\subsection{Weak coupling versus strong potential}\label{sec:strong-loop}

For high-fidelity state transfer, we compare the \trex method (weak edge couplings) 
with methods based on strong potential (large self-loop weights) (see \cite{ls25,ls24}).
We restate the main results in Lippner and Shi \cite{ls25}.
Let $G$ be a graph with maximum degree $m$ and let $a,b$ be a pair of vertices in $G$ of distance $d$ from
each other and with cospectrality $c$ (see Definition 2.6 in \cite{ls25}).
If we attach self-loops of weight $Q$ to both $a$ and $b$ where
\[
	Q > \frac{16}{\epsilon^{1/\alpha}}m^{1+\beta}
\]
with $\alpha = \min\{2,c-d+1\}$ and $\beta = \max\{1/2,d/(c-d+1)\}$, 
then there is state transfer from $a$ to $b$ with fidelity at least $1-\epsilon$ (see Theorem 1.4 in \cite{ls25})
in time $t_0$ (see Theorem 1.6 in \cite{ls25}) where
\[
	t_0 < 2\pi(Q+m)^{d-1}.
\]
Here, state transfer time scales exponentially with the distance $d$ (or diameter, in the worst case).
Intuitively, this method increases the potential on $a$ and $b$ in such a way that the graph $G$ becomes
negligible (which facilitates state transfer) but at the expense of the distance (of the order $\OO(1/\epsilon^d)$).
In contrast, the state transfer time for the \trex method depends mainly on the condition number of $A(G)$.
For example, for the Krawtchouk chain of length $n$, the condition number is $\OO(n)$ and for the unweighted path
$P_n$ the condition number is $\OO(n^2)$. Thus, the state transfer time of our scheme is {\em polynomial} in the
diameter.

\subsection{Transport in the Presence of Localization}

Anderson \cite{a58} discovered a localization phenomenon in light propagating through randomly defective crystals. 
In random matrix theory, this is modeled by Jacobi matrices under random diagonal perturbations.
The works of \cite{klmw07,redsm16} gave strong evidence that quantum state transfer on spin chains 
is adversely affected by such localization. 
Almeida \etal \cite{adl18} had studied the \trex method under a structured disorder model.

Our main result shows, at the cost of four protected vertices (see \cref{fig:anderson3}),
there is a protocol for high-fidelity state transfer on spin chains robust against {\em arbitrary} 
(in particular, asymmetric) random diagonal perturbations.

\begin{theorem} \label{thm:anderson3}
Let $H$ be a Jacobi matrix of order $n+2$ with $\norm{H}=1$. 
Let $D = \diag(\zeta_k)_{k=1}^{n+2}$ be a diagonal matrix where $\zeta_2,\ldots,\zeta_{n+1}$ are arbitrary random entries
and with $\zeta_1 = \zeta_{n+2} = 0$.
Let $H_1$ be a Jacobi matrix of order $n+4$ obtained from $H + D$ by attaching two pendant edges of weight $\delta$ each.
Suppose $a$ is the first vertex of $H$ (or the second vertex of $H_1$).
Then, there exists $B \in \RR$ so that high-fidelity quantum state transfer occurs on $H_1 + B\ketbra{a}{a}$
between antipodal vertices of $H_1$.

\begin{proof}
Let $a$ and $b$ be the first and the last vertices of $H$, respectively. We attach a self-loop of weight $B$ to $a$.
The cospectrality parameter (see \cref{def:cospectral}) is given by
\[
	\epsilon(B) 
	= \frac{\bra{a}H^{-1}\ket{a} - \bra{b}H^{-1}\ket{b}}{\bra{a}H^{-1}\ket{b}}
	= \frac{\det(H_{a,a}) - \det(H_{b,b})}{\det(H_{b,a})}
\]
where $H_{b,a}$ denotes the matrix obtained from $H$ by deleting row $b$ and column $a$
(see \cite{hj2}, Chapter 0, Sections 3.1 and 8.4). 
If $\omega_1,\ldots,\omega_{n+1}$ are the edge weights of $H$, then $\det(H_{b,a}) = \prod_{k=1}^{n+1} \omega_k$.
Then,
\[
	\epsilon(B) = (\prod_{k=1}^{n+1} w_k)^{-1} (\alpha + \beta B)
\]
where $\alpha$ and $\beta$ are functions of the unknown loop weights $\zeta_1,\ldots,\zeta_n$ 
and the known edge weights $w_1,\ldots,w_{n+1}$. So, $\epsilon(B)$ is linear in $B$.

For any given $H$, a value of $B=-\alpha/\beta$ exists that sets $\epsilon(B)=0$. Even if we don't know all the parameters of $H$, using \cref{thm:high-fidelity}, we could perform a state transfer experiment on $H$ to estimate $\epsilon(0)$ and $\epsilon(1)$,
which allows us to recover $\alpha$ and $\beta$. This implies antipodal state transfer with fidelity of $1-o_n(1)$ by \cref{thm:high-fidelity}.
\end{proof}
\end{theorem}

\begin{figure}[t]
\begin{center}
\begin{tikzpicture}[-,>=stealth',auto,node distance=1.75cm,thick,
	main node/.style={circle,circular drop shadow,draw,font=\bfseries}, 
	clean node/.style={circle,draw,font=\bfseries}, 
	main edge/.style={-,>=stealth'}]
\tikzset{every loop/.style={min distance=10mm, in=55, out=125}}
  \node[clean node] (1) {};
  \node[clean node] (2) [right of=1] {};
  \node[main node] (3) [right of=2] {};
  \node[main node] (4) [right of=3] {};
  \node[main node] (5) [right of=4] {};
  \node[main node] (6) [right of=5] {};
  \node[main node] (7) [right of=6] {};
  \node[clean node] (8) [right of=7] {};
  \node[clean node] (9) [right of=8] {};

  \path[shorten >=1pt]
    (1) edge node {$\delta$} (2)
    (2) edge node {} (3)
    (7) edge node {} (8)
    (8) edge node {$\delta$} (9);
  \path[main edge]
    (3) edge node {} (4)
    (4) edge node {} (5)
    (5) edge node {} (6)
    (6) edge node {} (7);

  \path[-]
    (2) edge [loop above] node {$B$} (2)
    (3) edge [loop above] node {$\zeta_1$} (3)
    (4) edge [loop above] node {$\zeta_2$} (4)
    (5) edge [loop above] node {$\ldots$} (5)
    (6) edge [loop above] node {$\zeta_{n-1}$} (6)
    (7) edge [loop above] node {$\zeta_n$} (7);
\end{tikzpicture}
\caption{State transfer protocol robust against Anderson localization:
the first and last two vertices (unshaded) are protected from noise 
while the middle vertices (shaded) are under the influence of Anderson localization.
The control parameter $B$ is used to ensure cospectrality.
}
\label{fig:anderson3}
\vspace{.2in}
\end{center}
\end{figure}
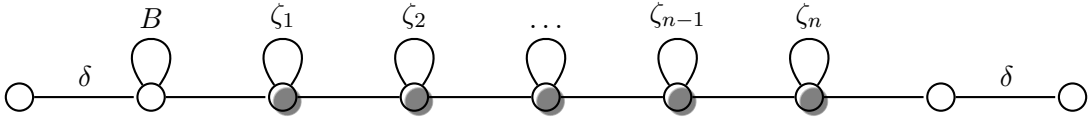

In \cref{fig:p57-anderson}, we illustrate the efficacy of the state transfer protocol 
used in \cref{thm:anderson3}.

\begin{figure}[h]
\includegraphics[scale=0.175]{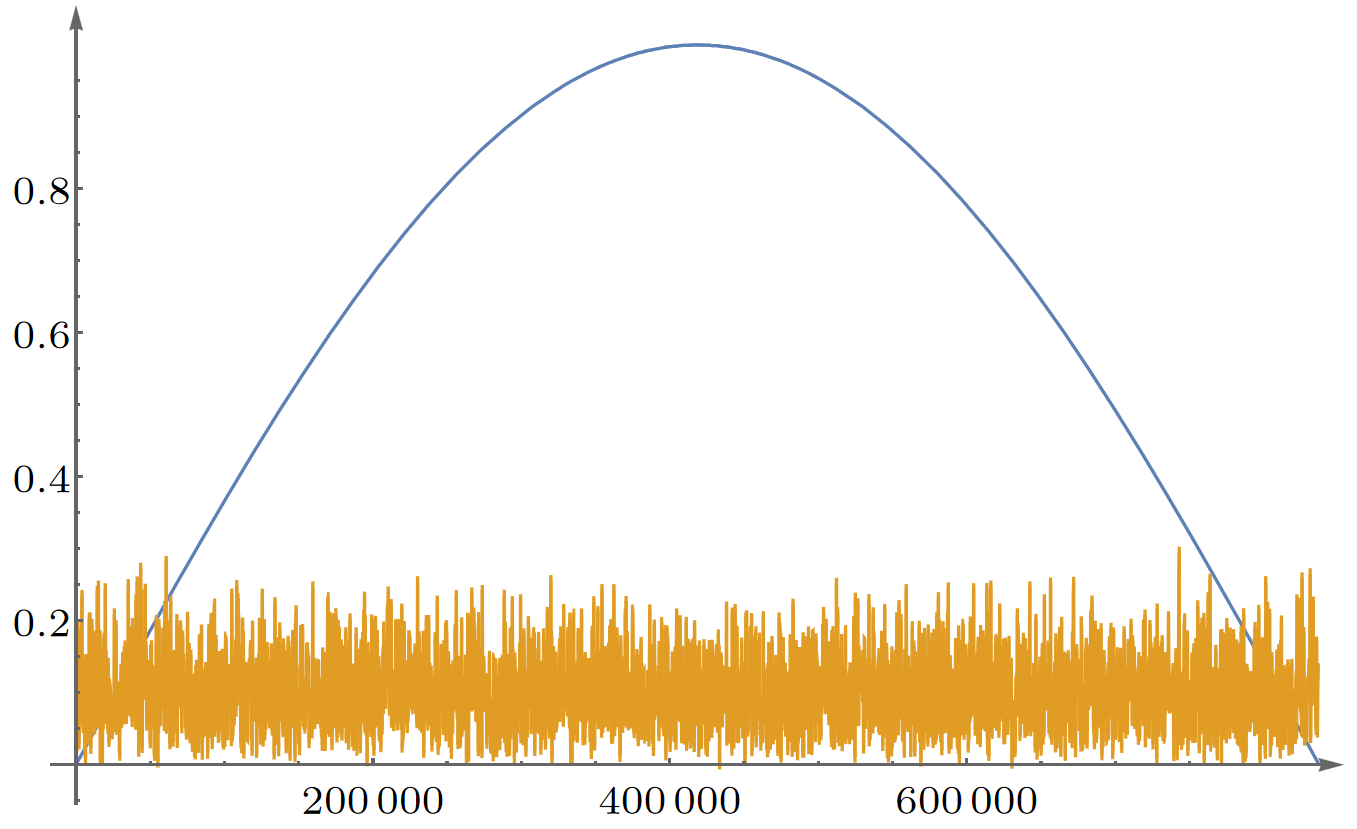}
\includegraphics[scale=0.2]{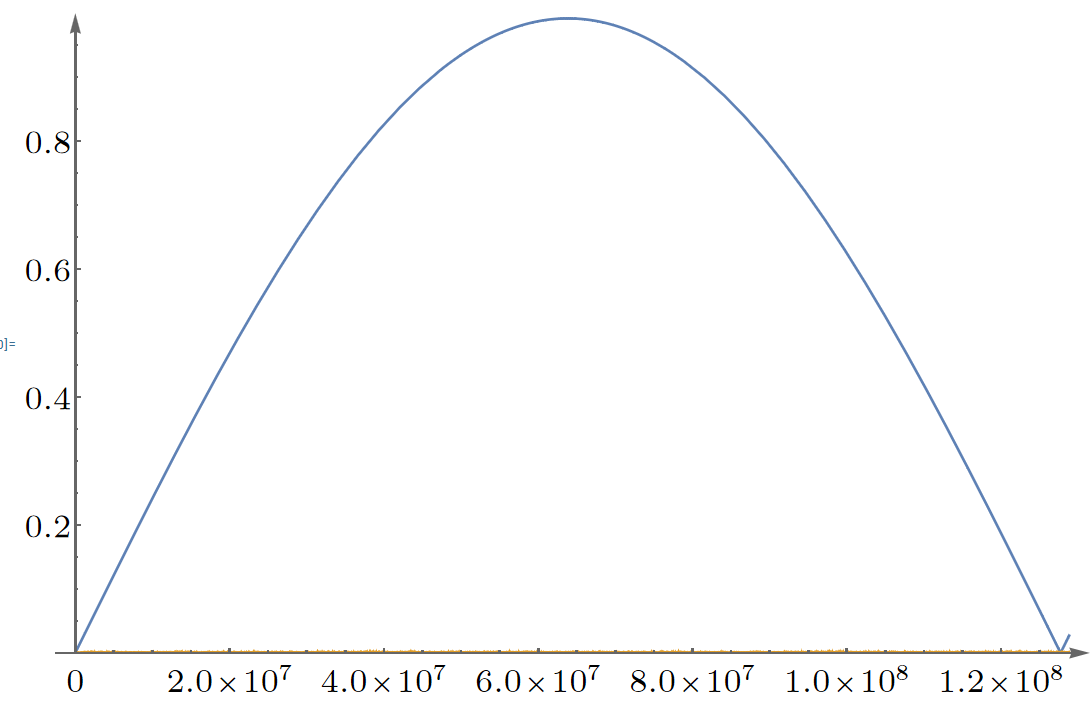}
\caption{The \trex protocol against Anderson localization on $P_{n}$.
Here, $n=55$ where we plot time $t$ (x-axis) versus antipodal fidelity (y-axis).
{\em Left}: \trex against Cauchy noise with parameter $0.06$ (as in \cite{klmw07}) with $\delta=0.002$ and fidelity $0.9997$.
{\em Right}: \trex against noise from $U(-2,2)$ with $\delta=0.0067$ and fidelity $0.992$ 
(localization practically destroys antipodal quantum transport).
}
\label{fig:p57-anderson}
\end{figure}

\newcommand{\xa}{{a}}
\newcommand{\xb}{{b}}
\section{Quantum Speedup for Hitting Times}\label{sec:hit}

We consider a random walk on an undirected connected graph $G$ which induces a finite irreducible and
reversible Markov chain with stationary distribution $\pi$ (see Aldous and Fill \cite{aldous-fill}, Chapter 4).
For a vertex $b$, let $T_b$ denote the number of steps until the random walk reaches $b$ for the first time.
Let $\Exp_a[T_b]$ denote the expected value of $T_b$ when we start the random walk at vertex $a$.
The average hitting time of $G$ is defined as
\begin{equation} \label{eqn:aht}
\tau_0 = \sum_{a,b} \pi_a \pi_b \Exp_a[T_b].
\end{equation} 
The maximal mean commute time of $G$ is defined as
\begin{equation} \label{eqn:mmct}
\tau^\star = \max_{a,b} (\Exp_a[T_b] + \Exp_b[T_a]).
\end{equation}

The values of $\tau^\star$ and $\tau_0$ for standard families of graphs are given in \cref{tab:hit_times}
(based on Chapter 5 in \cite{aldous-fill}).
We compare these classical notions of hitting times with the quantum hitting time obtained by attaching \trex arms 
to vertices $a$ and $b$ of the graph $G$. For a fixed threshold $\varrho$, we let $\tau_Q$ be the minimum $t$ 
so that the quantum walk satisfies $f^2_{a,b}(t) \ge \varrho$.

In \cref{tab:hit_times}, we observe quantum speedups in $\tau_Q$ compared to $\tau^\star$ (and $\tau_0$) 
on certain families of graphs.

\begin{table}[h]
\centering
\begin{tabular}{cccc}
\toprule
Graph 			& $\tau^\star$	& $\tau_0$ 	& $\tau_Q$ \\
\midrule
cycle 			& $N^2$			& $N^2$ 	& $N^{3/2}\delta^{-1}$ \\
path  			& $N^2$			& $N^2$ 	& $N^{3/2}\delta^{-1}$ \\
complete graph 	& $N$			& $N$ 		& $\sqrt{N}\delta^{-1}\ ^\star$ \\
star 			& $N$			& $N$ 		& $N\delta^{-2}\ ^\star$ \\
barbell 		& $N^3$			& $N^3$ 	& $N^{5/2}\delta^{-1}$ or $N\delta^{-2}$\\
lollipop 		& $N^3$			& $N^2$ 	& $N^{3/2}\delta^{-1}$ \\
necklace 		& $N^2$			& $N^2$ 	& $N^{3/2}\delta^{-1}$\\
$d$-cube ($d=\log_2 N$) & $N$	& $N$ & $\log N$ \\
dense regular graphs	& $N$	& $N$ & $\sqrt{N}\delta^{-1}\kappa^{-1}$ \\
rook's graph 	& $N$			& $N$ 		& $N\delta^{-2}\ ^\star$ \\
\bottomrule
\end{tabular}
\vspace{0.2in}
\caption{Comparison of classical ($\tau^\star,\tau_0$) and quantum hitting times using \trex ($\tau_Q$) for various graphs of size $N$. 
The bounds are asymptotic in $N$. The coupling parameter $\delta$ is $\order{1}$ but small.
$^\star$ indicates that further speed enhancements are available if $\delta$ is increased beyond $\order{1}$. 
}
\label{tab:hit_times}
\end{table}

\subsection{Resonant tunneling} \label{sec:resonant}
We describe {\em resonant tunneling}, a variation of the \trex method, which works (under mild assumptions) 
for singular base graphs.
Suppose the base graph $G_0$ has an adjacency matrix $H_0$ with zero as a {\em simple} eigenvalue and $\ket{\rho}$ as its corresponding eigenvector.
Consider the disjoint union $G = G_0 \cup \overline{K}_2$ and suppose $a$ and $b$ are the vertices of the coclique $\overline{K}_2$;
here, the adjacency matrix of $G$ is $H = H_0 \oplus O_2$.

Choose two vertices $\alpha$ and $\beta$ in $G_0$ and attach the \trex arms through the coupling
$W = \ketbra{a}{\alpha} + \ketbra{\alpha}{a} + \ketbra{b}{\beta} + \ketbra{\beta}{b}$.
Now, we consider the perturbed Hamiltonian 
\[
	H(\delta) = H + \delta W.
\]
Let $P_0 = \ketbra{a}{a}+\ketbra{b}{b}+\ketbra{\rho}{\rho}$ be the projection onto the zero eigenspace of $H$.
Suppose $P_0(\delta)$ is the eigenprojector of $H(\delta)$ corresponding to $P_0$ after perturbation.
Let $U$ be the direct rotation operator from $P_0(\delta)$ to $P_0$. This unitary operator is called 
the {\em Schrieffer-Wolff} transformation (see Bravyi \etal \cite{bdl11}, Definition 3.1).
In this framework, the effective Hamiltonian relative to $P_0$ is given by
\begin{equation} \label{eqn:heff-sw}
	\heff = P_0U(H + \delta W)U^\dagger P_0.
\end{equation}
If $\delta$ is sufficiently small and the spectral gap between $0$ and the other eigenvalues of $H_0$ are large enough, then
\[
	e^{-it\heff} = U^\dagger e^{-it(H + \delta W)}U.
\]
provided the initial and target states of interest are both in $P_0$. This is because $P_0(\delta)$ is an invariant
subspace of $H + \delta W$. Moreover, up to a first-order approximation of $U$, we get (see \cite{bdl11}, Equation 3.13)
\begin{equation} \label{eqn:heff-approx}
	\heff = P_0 H P_0 + \delta P_0 W P_0.
\end{equation}
In our \trex setting, under the $\{\ket{a},\ket{\rho},\ket{b}\}$ basis, we have
\[
	\heff =
	\epsilon\Delta
	\begin{pmatrix}
	0 & \braket{\alpha}{\rho} & 0 \\
	\braket{\alpha}{\rho} & 0 & \braket{\beta}{\rho} \\
	0 & \braket{\beta}{\rho} & 0
	\end{pmatrix}
\]
where $\Delta = \min\{|\lambda| : \lambda \in \Sp(G_0), \lambda \neq 0\}$ (which is included so that
the effective Hamiltonian has a bounded norm) and
$\epsilon$ is a sufficiently small constant.
Assuming $\braket{\alpha}{\rho} \neq 0$ and $\braket{\alpha}{\rho} = \braket{\beta}{\rho}$, 
we see that the quantum walk on $\heff$ maps $\ket{a}$ to $\ket{b}$ in time 
\[
	\frac{\pi}{\sqrt{2}}\frac{1}{\epsilon\Delta |\braket{\alpha}{\rho}|}.
\]

\par\noindent{\em Remark}:
The Laplacian $L = D-A$ of a connected undirected graph $G$ of order $N$ has zero as a simple eigenvalue 
with the normalized all-one vector as eigenvector. Thus, resonant tunnelling occurs in time 
$\OO(\sqrt{N}\epsilon^{-1}\Delta^{-1})$ where $\Delta = \lambda_2(G)$ is the algebraic connectivity of $G$.

\subsection{Examples}
In this section, we analyze a subset of families of graphs (as listed in \cref{tab:hit_times}) 
and derive their quantum hitting times under the \trex method. 
We omit the other families as the analyses are similar.

\subsubsection{Path}
The path $P_N$ has vertex set $\{1,\ldots,N\}$ with edges $(j,k)$ provided $|j-k|=1$.
The eigenvalues of $P_N$ are $\lambda_j = 2\cos(j\pi/(N+1))$ (of multiplicity one)
with eigenvector whose $k$-th entry is $\braket{k}{\lambda_j} = \sqrt{\frac{2}{N+1}}\sin(jk\pi/(N+1))$,
where $j,k=1,\ldots,N$.

If $N$ is odd, $\lambda_{k}=0$ for $k=(N+1)/2$ and the zero-eigenvector $\ket{\rho}$ satisfies
and $|\braket{1}{\rho}| = \sqrt{2/(N+1)} \sim 1/\sqrt{N}$.
Moreover, $\Delta = 2\sin(\pi/(N+1)) \sim 1/N$ as the closest eigenvalues to zero are 
$\lambda_{k \pm 1} = 2\cos((k \pm 1)\pi/(N+1)) = \mp 2\sin(\pi/(N+1))$.
By resonant tunneling, the quantum hitting time is $\order{N^{3/2}\delta^{-1}}$.
In \cite{wlkggb05}, a quantum hitting time that scales linearly with $N$ was claimed
but we were unable to corroborate this.

\subsubsection{Cycle}
The cycle $C_N$ has $\ZZ_N = \{0,1,\ldots,N-1\}$ as its vertices and $(k,k \pm 1\pmod{N})$ as its edges.
If $N$ is even, where $N = 2m$ for $m$ even, then the quotient of $C_N$ (modulo an equitable partition)
is a weighted path $\tilde{P}_{m+1}$ of odd length (which differs from the unweighted path $P_{m+1}$ by
the $\sqrt{2}$ weights on the pendant edges). As the quantum walks on $C_N$ and on $\tilde{P}_{m+1}$
are equivalent relative to end-to-end transfer, we may appeal to resonant tunneling for $P_{m+1}$.
So, the quantum hitting time is $\order{N^{3/2}\delta^{-1}}$.

\subsubsection{Clique}
The complete graph (or clique) $K_N$ contains all edges $(k,\ell)$, whenever $1 \le k \neq \ell \le N$.
The quotient of $K_N$ with cells $\{1\}$, $\{2,\ldots,N-1\}$ and $\{N\}$ is a weighted path given by
\begin{equation} \label{eqn:clique}
	J = 
	\begin{pmatrix}
	0 & \sqrt{N-2} & 0 \\
	\sqrt{N-2} & N-3 & \sqrt{N-2} \\
	0 & \sqrt{N-2} & 0
	\end{pmatrix}.
\end{equation}
The eigenvalues of $J$ are approximately $\sqrt{N-2}(c \pm d)$, where $c = \tfrac{1}{2}(N-3)/\sqrt{N-2}$ 
and $d = \sqrt{c^2 + 2}$, and $0$. The zero eigenvector $\ket{\rho}$ satisfies $|\braket{1}{\rho}| = 1/\sqrt{2}$. 
We apply resonant tunneling with $\Delta \sim |c - d| \sim 1/\sqrt{N}$ (after dividing by the spectral norm).
Therefore, the quantum hitting time is $\order{\sqrt{N}\delta^{-1}}$.

\subsubsection{Hypercube}
The $d$-cube $Q_d$ has vertex set $\ZZ_2^d$ (thus, $N=2^d$) and contains edges of the form $(\alpha,\alpha \oplus e_k)$, 
for $k=1,\ldots,d$, where $\oplus$ denotes the bitwise exclusive OR operation.
From Christandl \etal \cite{cdel04,christandl2005}, it is known that perfect state transfer occurs between 
antipodal vertices of $Q_d$ in time $\tfrac{\pi}{2}\log_2 N$ (without the \trex arms).
This result was already known in the earlier work of Moore and Russell; see the appendix in \cite{mr02}.
We obtain the same $\OO(\log_2 N)$ bound by using \trex arms.

\subsubsection{Rook's graph}
The rook's graph is defined as the Cartesian product $K_m \Box K_m$ (hence $N=m^2$).
It is a strongly regular graph whose eigenvalues are $2m-2$ (of multiplicity one), $m-2$ (of multiplicity $2m-2$), 
and $-2$ (of multiplicity $(m-1)^2$). For vertices $\alpha=(1,1)$ and $\beta=(m,m)$, we get
$|\bra{\beta}A(K_m \Box K_m)^{-1}\ket{\alpha}| \sim 1/m^2$. Hence, the quantum hitting time is 
$\order{N\delta^{-2}}$ which shows no asymptotic speedup.


\subsubsection{Barbell}\label{sec:barbell}
The barbell graph $B_N$ is obtained from the disjoint union $2K_N \cup P_N$ by identifying the leaf vertices of $P_N$
with two vertices chosen from distinct cliques. 
Let $\alpha$ and $\beta$ be two vertices from the two $K_N$s that are not attached to the path $P_N$.
The barbell has a quotient structure (see \cref{fig:barbell}) described by a weighted path of $N+4$ vertices, 
whose adjacency matrix is a Jacobi matrix $H$ given by
\begin{equation} \label{eqn:barbell-quotient}
	H = 
	\begin{pmatrix}
	0 			& \sqrt{N-2} 	& 0				&			&			&				& 				& \\
	\sqrt{N-2} 	& N-3 			& \sqrt{N-2} 	&			&			&				& 				& \\
	0			& \sqrt{N-2}	& 0				& 1			&			&				& 				& \\
				&				& 1				& 0			& 1			&				& 				& \\
	\vdots		& \vdots		& \vdots		& \vdots	& \vdots	& \vdots		& \vdots 		& \vdots \\
				&				&				& 1			& 0			& 1				&				& \\
				&				&				&			& 1			& 0				& \sqrt{N-2} 	& 0 \\
				&				&				&			&			& \sqrt{N-2}	& N-3			& \sqrt{N-2} \\
				&				&				&			&			& 0				& \sqrt{N-2}	& 0
	\end{pmatrix}
\end{equation}
To ensure bounded norm, we need to scale $H$ by $1/N$.
Depending on the parity of $N$, we show a comparison of the two \trex methods:
resonant tunneling versus Feshbach-Schur.

\vspace{0.1in}
\par\noindent{\em Case 1: $N$ is odd}.
Observe that $\tfrac{1}{N}H$ has zero as a simple eigenvalue with eigenvector $\ket{\rho}$ whose entries are given by 
\begin{equation} \label{eqn:resonant-entry}
	\braket{k}{\rho} = \sqrt{\tfrac{2}{N+5}} \times \iverson{\mbox{$k$ is odd}} \times (-1)^{\lfloor k/2\rfloor}
\end{equation}
for $k=1,\ldots,N+4$. This allows us to apply resonant tunneling.
To estimate the spectral gap to the zero eigenvalue, note the cliques at the end of the barbell
have similar structure to \cref{eqn:clique} which has distance $\kappa \sim 1/N$ to eigenvalue zero. 
But, the eigenvalues of the path $P_N$, after scaling by $1/N$, has gap $1/N^2$. This yields an effective Hamiltonian
\begin{equation} \label{eqn:heff-barbell-odd}
	\heff = 
	\frac{\delta}{N^2}
	\begin{pmatrix}
	0 & \braket{\alpha}{\rho} & 0 \\
	\braket{\alpha}{\rho} & 0 & \braket{\beta}{\rho} \\
	0 & \braket{\beta}{\rho} & 0
	\end{pmatrix}.
\end{equation}
As $\braket{\alpha}{\rho} = \braket{\beta}{\rho} = \sqrt{2/(N+5)} \sim 1/\sqrt{N}$ for $\alpha = 1$ and $\beta = N+4$ 
(by \cref{eqn:resonant-entry}), this gives a quantum hitting time of $\OO(N^{5/2}\delta^{-1})$.

\begin{figure}[t]
\begin{center}
\begin{tikzpicture}[-,>=stealth',auto,node distance=1.75cm,thick,
       main node/.style={circle,circular drop shadow,draw,font=\bfseries}, main edge/.style={-,>=stealth'}]
\tikzset{every loop/.style={min distance=10mm, in=55, out=125}}
  \node[main node] (1) {};
  \node[main node] (2) [right of=1] {};
  \node[main node] (3) [right of=2] {};
  \node[main node] (4) [right of=3] {};
  \node[main node] (5) [right of=4] {};
  \node[main node] (6) [right of=5] {};
  \node[main node] (7) [right of=6] {};
  \node[main node] (8) [right of=7] {};
  \node[main node] (9) [right of=8] {};

  \path[shorten >=1pt]
    (1) edge node {$a$} (2)
    (2) edge node {$a$} (3)
    (7) edge node {$a$} (8)
    (8) edge node {$a$} (9);
  \path[main edge]
    (3) edge node {$c$} (4)
    (4) edge node {$c$} (5)
    (5) edge node {$c$} (6)
    (6) edge node {$c$} (7);

  \path[-]
    (2) edge [loop above] node {$b$} (2)
    (8) edge [loop above] node {$b$} (8);
\end{tikzpicture}
\caption{The quotient structure of the barbell graph.
The unnormalized barbell has $a=\sqrt{N-2}$, $b=N-3$, and $c=1$.
}
\label{fig:barbell}
\vspace{.2in}
\hrule
\end{center}
\end{figure}

\ignore{
\begin{figure}[b]
\centering
\drawchainwithfields[0.45\textwidth]{$\sqrt{N-2}$,$\sqrt{N-2}$,1,1,1,1,$\sqrt{N-2}$,$\sqrt{N-2}$}{,$N-1$,,,,,,$N-1$,}
\caption{Quotient structure of the barbell on $3N$ vertices: two cliques $K_N$ joined by a path $P_N$ 
(intermediate vertices not shown).}\label{fig:barbell}
\end{figure}
}

\vspace{0.1in}
\par\noindent{\em Case 2: $N$ is even}.
We apply the \trex method using the Feshbach-Schur map.
Consider the Jacobi matrix $J$, parametrized by $a$, $b$, and $c$, given by (see \cref{fig:barbell})
\begin{equation} \label{eqn:barbell-jacobi}
	J = J(a,b,c) = 
	\begin{pmatrix}
	0 & a &   &   &   &   &   &   &  \\
	a & b & a &   &   &   &   &   &  \\
	  & a & 0 & c &   &   &   &   &  \\
	  &   & c & 0 & c &   &   &   &  \\
	  &   & \vdots &  & \vdots & \vdots & & \vdots & \\
	  &   &   &   &   & 0 & c &   & \\
	  &   &   &   &   & c & 0 & a & \\
	  &   &   &   &   &   & a & b & a \\
	  &   &   &   &   &   &   & a & 0 \\
	\end{pmatrix}
\end{equation}
Recall $J^{-1} = \Adj(J)/\det(J)$, where $\Adj(J)$ is the adjugate of $J$ (see \cite{hj2}, Section 0.8.2).
Therefore, we have $\bra{\alpha}J^{-1}\ket{\beta} = \det(J_{\beta,\alpha})/\det(J)$,
where $J_{\beta,\alpha}$ is $J$ with row $\beta$ and column $\alpha$ removed. 
By direct calculation, we see that $\det(J_{\beta,\alpha})=a^4 c^{N-1}$ and $\det(J)=\pm a^4c^N$.
Therefore,
\[
	|\bra{\alpha}J^{-1}\ket{\beta}| = c^{-1}.
\]
A similar calculation shows $\det(J_{\alpha,\alpha})=\det(J_{\beta,\beta})=a^2 c^N$, which implies
\[
	|\bra{\alpha}J^{-1}\ket{\alpha}| = |\bra{\beta}J^{-1}\ket{\beta}| = \frac{1}{a^{2}}.
\]
The effective Hamiltonian is thus given by
\begin{equation} \label{eqn:heff-barbell-even}
	\heff = 
	\frac{\delta^2}{\kappa}
	\begin{pmatrix}
	\bra{\alpha}J^{-1}\ket{\alpha} & \bra{\alpha}J^{-1}\ket{\beta} \\
	\bra{\beta}J^{-1}\ket{\alpha} & \bra{\beta}J^{-1}\ket{\beta}
	\end{pmatrix}
	=
	\frac{\delta^2}{\kappa}
	\begin{pmatrix}
	a^{-2} & c^{-1} \\
	c^{-1} & a^{-2}
	\end{pmatrix}.
\end{equation}
After a diagonal shift (which can be done without loss of generality), we have $\heff = (\delta^2/c\kappa) X$.
For the normalized barbell, we have $a=1/\sqrt{N}$, $b=1$, $c=1/N$.
Furthermore, the normalized barbell has condition number $\kappa \sim \OO(N^2)$.
Thus, we obtain a quantum hitting time of $\OO(N\delta^{-2})$ (which is nearly optimal due to
Lieb-Robinson bounds \cite{bravyi2006,nachtergaele2010}).

\subsection{Quantum Search on Edge}

In quantum spatial search (see \cite{fg98,cg03}), given a graph $G=(V,E)$ with adjacency matrix $A$,
we ask if there is $\gamma$ (global scaling) so that for any vertex $\alpha \in V$ (oracle location), 
we have $\tau$ (search time) such that
\[
	|\bra{\alpha}e^{-i\tau(\gamma A + \ketbra{\alpha}{\alpha})}\ket{\rho}| \ge 1-\epsilon
\]
where $\ket{\rho}$ is the normalized Perron eigenvector of $G$.
This captures an analog variant of Grover search on arbitrary graphs. 
The quadratic speedup requirement is expressed as $\tau = \OO(1/|\braket{\alpha}{\rho}|)$.
Typically, we require $\epsilon = o(1)$ but it can be relaxed to a bounded constant.
In this framework, spatial search can be viewed as reverse mixing (as we transform $\ket{\rho}$ to $\ket{\alpha}$)
under the rank-one perturbation of $A$ (induced by the oracle position).

Imagine a new framework where the oracle, instead of being a strong potential located at a vertex,
is now a pendant edge weakly coupled to $G$. 
To find this pendant edge (or oracle), we randomly place another pendant edge on the graph 
(which, in the vertex model, is the Random Target Lemma \cite{lpw}).
Next, we appeal to the \trex argument to claim a high-fidelity state transfer between the two pendant edges.
For families of vertex-transitive graphs, we may use the hitting times given in Table \ref{tab:hit_times}.
In contrast, the same idea but with strong vertex potentials will run into the same issues as outlined
in \cref{sec:strong-loop}.

\section{Conclusions}

The \trex idea is a paradoxical, yet powerful, method for solving quantum transport problems.
We gave a mathematical treatment of the \trex idea using the Feshbach-Schur method from perturbation theory \cite{gs20}. 
A summary of our results and related open questions is given in the following:
\begin{enumerate}
\item An elementary and rigorous proof that the \trex method yields high-fidelity state transfer.\\
	The requirements on the base graph is modest: it should have good condition number and its two vertices involved
	in the state transfer must be $\gamma$-cospectral.
	In contrast to perfect and pretty-good state transfer (see \cite{cg}), the latter condition is less stringent
	than {\em strong} cospectrality. 
    Also, see \cite{cgs} for a related work on peak state transfer.

\item The transfer time of the \trex protocol does not depend directly on the diameter.\\
	In contrast, the transfer time of protocols based on strong potentials increase exponentially 
	with the diameter (see \cite{ls24,ls25}). 
	Is there a speed limit for high-fidelity state transfer?
	 
\item A high-fidelity state transfer protocol robust against localization in spin chains.\\
	The key idea was to exploit nonsingularity of the underlying spin chains along with a few protected vertices.
	Our protocol was able to overcome {\em arbitrary} noise (while possibly trading off state transfer time), in sharp contrast to the negative results in \cite{klmw07}.

\item A simple method to find speedups in hitting times.\\
	The versatility of the \trex method was evident from the variety of graphs that were analyzed 
	(thereby avoiding {\em ad hoc} methods).
	We then derived a quantum search algorithm to locate an oracle hiding as a pendant edge. 
\end{enumerate}
To complement the above list, we mention other related results inspired by the \trex paradigm:
\begin{enumerate}[resume]
\item A simpler matrix inversion quantum algorithm \cite{kt}.\\
	This is based on attaching multiple \trex arms to a nonsingular base graph.
\item An optimal PST protocol robust against timing sensitivities on spin chains \cite{kkt}.\\
	This exploits a link between weak couplings and the second derivative of the PST curve.
\end{enumerate}
We strongly suspect there are further uses of the \trex method waiting to be unearthed.

\section*{Acknowledgments}

C.T. is supported by NSF grant OSI-2427020.
We thank Ada Chan for bringing the Random Target Lemma to our attention and
Sooyeong Kim for helpful discussions on the barbell graphs.



\begin{thebibliography}{1}

\bibitem{acde04}
C. Albanese, M. Christandl, N. Datta, A. Ekert,
Mirror Inversion of Quantum States in Linear Registers,
\textit{Physical Review Letters} {\bf 93}:230502, 2004.

\bibitem{aldous-fill}
D. Aldous, J. Fill,
Reversible Markov Chains and Random Walk on Graphs,
Unfinished monograph, recompiled 2014, available at \url{http://www.stat.berkeley.edu/$\sim$aldous/RWG/book.html}.

\bibitem{adl18}
G. Almeida, F. de Moura, M. Lyra,
High-fidelity state transfer through long-range correlated disordered quantum channels,
\textit{Physics Letters A} {\bf 382}:1335, 2018.

\bibitem{a58}
P. Anderson,
Absence of Diffusion in Certain Random Lattice Models,
\textit{Physical Review} {\bf 109}(5):1492-1505, 1958.

\bibitem{bose2003}
S. Bose,
Quantum Communication through an Unmodulated Spin Chain,
\textit{Physical Review Letters} {\bf 91}:207901, 2003.

\bibitem{bdl11}
S. Bravyi, D. DiVincenzo, D. Loss,
Schrieffer-Wolff transformations for quantum many-body systems,
\textit{Annals of Physics} {\bf 326}:2793-2826, 2011.

\bibitem{bravyi2006}
S. Bravyi, M. Hastings, F. Verstraete,
Lieb-Robinson Bounds and the Generation of Correlations and Topological Quantum Order,
\textit{Physical Review Letters} {\bf 97}:050401, 2006.

\bibitem{cmf16}
X. Chen, R. Mereau, D. Feder,
Asymptotically perfect efficient quantum state transfer across uniform chains with two impurities,
\textit{Physical Review A} {\bf 93}:012343, 2016.

\bibitem{ccdfgs03}
A.~Childs, R.~Cleve, E.~Deotto, E.~Farhi, S.~Gutmann, and D.~Spielman,
{Exponential Algorithmic Speedup by a Quantum Walk},
\textit{Proceedings of the thirty-fifth annual ACM symposium on Theory of computing}, p59-68, 2003.

\bibitem{cg03}
A. Childs, J. Goldstone,
Spatial search by quantum walk,
\textit{Physical Review A} {\bf 70}:022314, 2003.

\bibitem{cdel04}
M. Christandl, N. Datta, A. Ekert, A. Landahl,
Perfect State Transfer in Quantum Spin Networks,
\textit{Physical Review Letters} {\bf 92}(18):187902, 2004.

\bibitem{christandl2005}
M. Christandl, N. Datta, T. Dorlas, A. Ekert, A. Kay, A. Landahl,
Perfect Transfer of Arbitrary States in Quantum Spin Networks,
\textit{Physical Review A} {\bf 71}(3):032312, 2005.

\bibitem{cg}
G. Coutinho, C. Godsil,
\textit{Graph Spectra and Continuous Quantum Walks},
book draft.

\bibitem{cgs}
G. Coutinho, K. Guo, V. Schmeits,
Peak state transfer in continuous quantum walks,
\href{https://arxiv.org/abs/2505.11986}{arxiv:2505.11986}.

\bibitem{dss21}
G. Dusson, I. Sigal, B. Stamm,
The Feshbach-Schur Map and Perturbation Theory,
\textit{Partial Differential Equations, Spectral Theory, and Mathematical Physics:
The Ari Laptev Anniversary Volume}, p65-88, EMS, 2021.

\bibitem{fg98}
E. Farhi, S. Gutmann,
Analog analogue of a digital quantum computation,
\textit{Physical Review A} {\bf 57}(4):2403, 1998.

\bibitem{fg98b}
E. Farhi, S. Gutmann,
Quantum Computation and Decision Trees,
\textit{Physical Review A} {\bf 58}(2):915, 1998.

\bibitem{gkss12}
C. Godsil, S. Kirkland, S. Severini, J. Smith,
Number-Theoretic Nature of Communication in Quantum Spin Systems,
\textit{Physical Review Letters} {\bf 109}:050502, 2012.

\bibitem{gs20}
S. Gustafson, I. Sigal,
\textit{Mathematical Concepts of Quantum Mechanics}, 3rd ed.,
Springer, 2020.

\bibitem{hj2}
R. Horn, C. Johnson,
\textit{Matrix Analysis}, 2nd ed.,
Cambridge University Press, 2013.

\bibitem{ks05}
P. Karbach, J. Stolze,
Spin chains as perfect quantum state mirrors,
\textit{Physical Review A} {\bf 72}:030301, 2005.

\bibitem{k80}
T. Kato,
\textit{Perturbation Theory for Linear Operators},
Springer, 1980.

\bibitem{k10}
A. Kay,
A Review of Perfect State Transfer and its Application as a Constructive Tool,
\textit{International Journal of Quantum Information} {\bf 8}:641, 2010.

\bibitem{kkt}
A. Kay, S. Kim, C. Tamon,
Optimising Perfect Quantum State Transfer for Timing Insensitivities,
\href{https://arxiv.org/abs/2507.18872}{arxiv:2507.18872}.

\bibitem{kt}
A. Kay, C. Tamon,
Matrix Inversion by Quantum Walk,
\href{https://arxiv.org/abs/2508.06611}{arxiv:2508.06611}.

\bibitem{klmw07}
J.P. Keating, N. Linden, J.C.F. Matthews, A. Winter,
Localization and its consequences for quantum walk algorithms and quantum communication,
\textit{Physical Review A} {\bf 76}:012315, 2007.

\bibitem{kempe2002}
J. Kempe,
Discrete Quantum Walks Hit Exponentially Faster,
\textit{Probability Theory and Related Fields} {\bf 133}:215-235, 2005.

\bibitem{lpw}
D. Levin, Y. Peres, E. Wilmer,
\textit{Markov Chains and Mixing Times},
American Mathematical Society, 2009.

\bibitem{ls24}
G. Lippner, Y. Shi,
Quantifying state transfer strength on graphs with involution,
\textit{Quantum Information Processing} {\bf 23}:166, 2024.

\bibitem{ls25}
G. Lippner, Y. Shi,
Strong quantum state transfer on graphs via loop edges,
\textit{Linear Algebra and Its Applications} {\bf 704}:77-91, 2025.

\bibitem{mr02}
C. Moore, A. Russell,
Quantum walks on the Hypercube,
\textit{Proc. 6th International Workshop on Randomization and Approximation Techniques in Computer Science
(RANDOM 2002)}, 164-178, 2002.

\bibitem{nachtergaele2010}
B. Nachtergaele, R. Sims,
Lieb-Robinson Bounds in Quantum Many-Body Physics,
``Entropy and the Quantum,'' Robert Sims and Daniel Ueltschi (Eds), 
\textit{Contemporary Mathematics}, volume 529, American Mathematical Society, p141-176, 2010.

\bibitem{redsm16}
R. Ronke, M. Estarellas, I. D'Amico, T. Spiller, T. Miyadera,
Anderson localisation in spin chains for perfect state transfer,
\textit{Eur. Phys. J. D} {\bf 70}:189, 2016.

\bibitem{sz07}
J. Sjostrand, M. Zworski,
Elementary Linear Algebra for Advanced Spectral Problems,
\textit{Annales de L'Institut Fourier} {\bf 57}(7):2095-2141, 2007.

\bibitem{t14}
G. Teschl,
\textit{Mathematical Methods in Quantum Mechanics}, 2nd ed.,
American Mathematical Society, 2014.

\bibitem{vz12}
L. Vinet, A. Zhedanov,
Almost perfect state transfer in quantum spin chains,
\textit{Physical Review A} {\bf 86}:052319, 2012.

\bibitem{vz24}
L. Vinet, A. Zhedanov,
Quantum state transfer with sufficient fidelity,
\href{https://arxiv.org/abs/2412.02321}{arxiv:2412.02321}.

\bibitem{wlkggb05}
A. W\'{o}jcik, T. \L{}uczak, P. Kurzy\'{n}ski, A. Grudka, T. Gdala, M. Bednarska,
\textit{Unmodulated spin chains as universal quantum wires},
Physical Review A {\bf 72}:034303, 2005.

\bibitem{wlkggb07}
A. W\'{o}jcik, T. \L{}uczak, P. Kurzy\'{n}ski, A. Grudka, T. Gdala, M. Bednarska,
\textit{Multiuser quantum communication networks},
Physical Review A, {\bf 75}:022330, 2007.

\end{thebibliography}
\end{document}